\documentclass{amsart}

\usepackage{bbm}

\usepackage{pgfplots}
\usepackage{subfigure}
\pgfplotsset{compat=1.13}

\def\C{\mathbb{C}}
\def\N{\mathbb{N}}
\def\E{\mathbb{E}}
\def\P{\mathbb{P}}
\def\R{\mathbb{R}}
\def\Z{\mathbb{Z}}
\def\eps{\varepsilon}
\newcommand{\diff}{\mathop{}\mathopen{}\mathrm{d}}
\newcommand\croc[1]{\left\langle #1\right\rangle}
\newcommand{\ind}[1]{\ensuremath{\mathbbm{1}_{\left\{#1\right\}}}}
\newcommand{\Var}{\mathop{}\mathopen{}\mathrm{Var}}

\def\cal{\mathcal}
\newtheorem{lemma}{Lemma}
\newtheorem{definition}{Definition}
\newtheorem{proposition}{Proposition}
\newtheorem{theorem}{Theorem}
\newtheorem{corollary}{Corollary}

\title{Allocation Schemes of Resources with Downgrading}
\date{Version of \today}

\author[C. Fricker]{Christine Fricker}
\author[F. Guillemin]{Fabrice Guillemin}
\address[F. Guillemin]{CNC/NCA Orange Labs2, Avenue Pierre Marzin, 22300 Lannion, France}
\email{Fabrice.Guillemin@orange.com}
\author[Ph. Robert]{Philippe Robert}
\author[G. Thompson]{Guilherme Thompson${ }^\dag$}\thanks{$\dag$ G. Thompson's research was supported by Brazilian Government/CAPES grant BEX 13748-13-0}
\address[C. Fricker, Ph. Robert, G. Thompson]{INRIA Paris, 2 rue Simone Iff, CS 42112,
75589 Paris Cedex 12, France}
\email{Philippe.Robert@inria.fr}
\urladdr{http://team.inria.fr/rap/robert}


\begin{document}

\begin{abstract}
We consider a server with large capacity delivering video files encoded in various resolutions. We assume that the system is under saturation in the sense that the total demand exceeds the server capacity $C$. In such  case, requests may be rejected.  For the policies considered  in this paper,  instead of rejecting a video request, it is downgraded. When the occupancy of the server is above some value $C_0{<}C$, the server delivers the video at a minimal bit rate.  The quantity $C_0$ is the bit rate adaptation threshold. For these policies, request blocking is thus replaced with bit rate adaptation. Under the assumptions of Poisson request arrivals and exponential service times, we show that, by rescaling the system,  a process associated with the occupancy of the server converges to some limiting process whose invariant distribution is computed explicitly.  This allows us to derive an asymptotic expression of the key performance measure of such a policy,  namely the equilibrium probability that a request is transmitted at requested bitrate. Numerical applications of these results are presented. 
\end{abstract}

\maketitle

\hrule

\vspace{-3mm}

\tableofcontents

\vspace{-10mm}

\hrule

\bigskip

\section{Introduction}
Video streaming applications have become over the past few years the dominant applications in the Internet and generate the prevalent part of traffic in today's IP networks; see for instance Guillemin et al.~\cite{icc13} for an illustration of the application breakdown in a commercial IP backbone network. Video files are currently downloaded  by customers from large data centers, like Google's data centers for YouTube files.   In the future,  it is very likely that video files will be delivered by smaller data centers located closer to end users, for instance cache servers disseminated in a national  network. It is worth noting that as shown in  Guillemin et al.~\cite{itc25}, caching is a very efficient solution for YouTube traffic. While this solution can improve performances by reducing delays, the limited capacity of those servers in terms of bandwidth and computing  can cause overload.

One possibility to reduce overload  is to use  bit rate adaptation. Video files can indeed be encoded at various bit rates (e.g, small and high definition video). If a node cannot serve a file at a high bit rate, then the video can be transmitted at a smaller rate. It is remarkable that video bit rate adaptation has become very popular in the past few years with the specification of MPEG-DASH standard where it is possible to downgrade the quality of a given transmission, see Schwarz et al.~\cite{Schwarz}, Sieber et al.~\cite{sieber},  A\~{n}orga et al.~\cite{Anorga2015}, Vadlakonda et al.~\cite{vadlakonda} and Fricker et al.~\cite{MAMA}. Adaptive streaming is also frequently used in mobile networks where bandwidth is highly varying.  In this paper, we investigate the effect of bit rate adaptation in a node under saturation.

\subsection*{Downgrading Policy}
We assume that customers  request video files encoded at various rates, say, $A_j$ for $j{=}1, \ldots, J$, with $1{=}A_1{<} A_2{<} \cdots{<}A_J$.  Jobs of class $j{\in}\{1,\ldots,J\}$ require bit rate $A_j$. The total capacity of the communication  link is $C$. If $\ell{=}(\ell_j)$ is the state of the network  at some moment, with $\ell_j$ being the number of class $j$ jobs,  the quantity $\croc{A,\ell}{=}A_1\ell_1{+}\cdots{+}A_J\ell_J$ has to be less than $C$.  The quantity $\croc{A,\ell}$ is defined as the {\em occupancy} of the link. The algorithm  has a parameter $C_0{<}C$ and works  as follows: If there is an arrival of a job of class $1{\leq}j_0{\leq}J$,
\begin{itemize}
\item if $\croc{A,\ell}{<} C_0$ then the job is accepted;
\item  if $C_0{\leq} \croc{A,\ell}{<}C$ then the job is accepted but as a class $1$ job, i.e.  it has an allocated bit rate of $A_1{=}1$ and service rate $\mu_1$;
 \item  if $\croc{A,\ell}{=} C$, the job is rejected. 
\end{itemize}
For $1{\leq}j{\leq}J$, jobs of class $j$ arrive according to a Poisson process with rate $\lambda_j$ and have an exponentially distributed transmission time with rate $\mu_j$. Additionally, it is assumed that
\[
\mu_1{\leq}\min(\mu_j,2{\leq} j{\leq} J).
\]

\subsection*{A Scaling Approach}
To study this allocation scheme, a scaling approach is used.  It is assumed that the server capacity is very large, namely scaled up by a factor $N$. The bit rate adaptation threshold and the request arrival rates are scaled up accordingly, i.e.
\begin{equation}\label{typscal}
\begin{cases}
  \lambda_j \mapsto \lambda_j N, \quad 1{\leq} j{\leq} J,\\
  C_0\mapsto c_0N \text{ \rm and }   C\mapsto c N.
\end{cases}
\end{equation}
{\em Performances of the algorithm.} Our main result shows that,  for the downgrading policy and if $c_0$ is  chosen conveniently, then
\begin{enumerate}
\item the equilibrium probability of rejecting a job converges to $0$ as $N$ goes to infinity;
\item the equilibrium probability of accepting a job without downgrading it converges to
  \[
\pi^-\stackrel{\text{\rm def.}}{=}\left(c_0 \mu_1{-}\sum_{j=1}^J\lambda_j\right)\left/\left(\mu_1\sum_{j=1}^J\frac{\lambda_j}{\mu_j}A_j{-}\sum_{j=1}^J\lambda_j\right.\right), 
  \]
   as $N$ goes to infinity. See  Theorem~\ref{theoeq} and Corollary~\ref{corol}. 
\end{enumerate}
The above formula gives an explicit expression of the success rate of this allocation mechanism. The quantity  $1{-}\pi^-$, the probability of downgrading requests,  can be seen as the ``price'' of the algorithm to avoid rejecting jobs. 

The scaling~\eqref{typscal} has been  introduced by Kelly to study loss networks. See Kelly~\cite{Kelly}. The transient behavior of these networks under this scaling has been analyzed by Hunt and Kurtz~\cite{Hunt}. This last reference provides essentially a framework to establish convenient convergence theorems  involving stochastic averaging principles.   This line of research has been developed in the 1990's to study uncontrolled  loss networks where a request is rejected as soon as its demand cannot be accepted. 

When the demand can be adapted to the state of the network,  for controlled loss networks, several (scarce) examples have been also analyzed during that period of time.  One can mention Bean et al.~\cite{Bean,Bean2},  Zachary and Ziedins~\cite{Zachary} and  Zachary~\cite{Zachary2} for example.   Our model can be seen as a ``controlled'' loss networks instead of a pure loss network.  Controlled loss networks may have mechanisms such as trunk reservation or  may  allocate requests  according to some complicated schemes depending on the state of the network.  In our case, the capacity requirements of requests are modified when the network is in a ``congested'' state.  

Contrary to classical uncontrolled loss networks, as it will be seen, the  Markov process associated to the evolution of the vector of the number of jobs for each class  is not reversible. Additionally, the invariant distribution of this process does not seem to have  a closed form expression. Kelly's approach~\cite{Kelly:2}  is based on an optimization problem, it cannot  be used in our case to get an asymptotic expression of some characteristics at equilibrium.   For this reason, the equilibrium behavior of these policies is investigated in a two step process:
\begin{enumerate}
\item  Transient Analysis.  We investigate the asymptotic behavior of some characteristics of the process on a finite time interval when the scaling parameter $N$ goes to infinity.
\item Equilibrium. The  stability properties of the limiting process are analyzed, we  prove that the equilibrium of the system for a fixed $N$ converges to the equilibrium of the limiting process.
\end{enumerate}
For our model, the transient analysis involves the {\em explicit} representation of the invariant distribution of a  specific class of Markov processes.  It is obtained with complex analysis arguments. As it will be seen, this representation plays an important role in the analysis of the asymptotic behavior at equilibrium. 

It should be noted that related models have recently been introduced to investigate resource allocation in a cloud computing environment where the nodes receive requests of several types of resources. We believe that this domain will receive a renewed attention in the coming years. See Stolyar~\cite{Stolyar2,Stolyar} and Fricker et al.~\cite{FGRT} for example.  In some way one could say that the loss networks are back and this is also a motivation of this paper to shed some light on the methods that can be used to study these systems. 

\subsection*{Outline of the paper}
We consider a system in overload. Because of  bit rate adaptation, requests may be downgraded but not systematically rejected as in a pure loss system.    As it will be seen, the  stability properties of this algorithm are linked to the behavior of a Markov process associated to the occupation of the link. Under exponential assumptions for  inter-arrival and  service times,  this process turns out to be, after rescaling by a large parameter $N$, a bilateral random walk instead of a reflected random walk as in the case of loss networks.  Using complex analysis methods, an explicit expression of the invariant distribution of this random walk is obtained.   With this result, the asymptotic expression of the probability  that, at equilibrium, a job is transmitted at its  requested rate (and therefore does not experience a bit rate adaptation)  is derived. 

This paper is organized as follows: In Section~\ref{Model}, we present the model used to study the network under some saturation condition. Convergence results when the scaling factor $N$ tends to infinity are proved in Section~\ref{SecMod}.  The invariant distribution of a limiting process associated to the occupation of the link  is computed in Section~\ref{SecInv} by means of complex analysis techniques. Applications are discussed in Section~\ref{App}.

 \medskip
\noindent
{\bf Acknowledgments}\\
The authors are very grateful to an anonymous referee for pointing out a gap in the proof of Theorem~\ref{theoeq} in the first version of this work. 
\section{Model description}\label{Model}
One considers a service system where $J$ classes of requests arrive at a server with bandwidth/capacity $C$. Requests of class $j$, $1{\leq} j{\leq} J$, arrive according to a Poisson process ${\cal N}_{\lambda_j}$ with rate $\lambda_j$. A class $j$ request  has a bandwidth requirement of $A_j$ units for a duration of time which is exponentially distributed with parameter $\mu_j$. For the systems investigated in this paper, there is no buffering, requests have to be processed at their arrival otherwise they are rejected. Without any flexibility on the resource allocation, this is a classical loss network with one link. See Kelly~\cite{Kelly} for example.  

This paper investigates  allocation schemes  which  consist of reducing the bandwidth  allocation of arriving requests to a minimal value  when the link has a high level of congestion.  In other words the service is downgraded for new requests arriving  during a saturation phase.  If the system is correctly designed,  it will reduce  significantly the fraction of rejected transmissions and, hopefully, few jobs will in fact  experience  downgrading.

\subsection{Downgrading policy ${\cal D}(C_0)$}
We introduce $C_0{<}C$, the parameter $C_0$ will indicate the level of congestion of the link.  It is assumed that the vector of integers $A{=}(A_j)$ is such that $A_1{=}1{<} A_2{<} \cdots{<} A_J$.  The condition $A_1{=}1$ is used to simplify the presentation of the results and to avoid problems of irreducibility in particular but this is not essential. 

If the network is in state $\ell{=}(\ell_j)$ and if the  occupancy $\croc{A,\ell}$  is  less than $C_0$, then any arriving request is accepted. If  the occupancy is between $C_0$ and $C{-}1$, it is accepted but with a minimal allocation,  as a class $1$ job. Finally it is rejected if the link is fully occupied, i.e. $\croc{A,\ell}{=}C$. It is assumed that $\mu_1\leq \mu_j$, for  $1{\leq}j{\leq}J$,  i.e. class~1 jobs are served with the smallest service rate.  

Mathematically, the stochastic model is close to a loss network with the restriction that a job may change its requirements depending on the state of the network. This is a controlled loss network,  see Zachary and Ziedins~\cite{Zachary2}.  It does not seem that, like in uncontrolled loss networks, the associated Markov process giving the evolution of the vector $\ell$  has reversibility properties, or that its invariant distribution has a product form expression.  Related schemes with product form are trunk reservation policies for which  requests of a subset of classes  are systematically rejected  when the level of congestion of the link is above some threshold.  See Bean et al.~\cite{Bean} and Zachary and Ziedins~\cite{Zachary} for example.  Concerning controlled loss networks, mathematical results are more scarce. One can mention networks where jobs requiring congested links are redirected to less loaded links. Several mathematical approximations have been proposed to study these models.  See  the surveys Kelly~\cite{Kelly} and Zachary and Ziedins~\cite{Zachary2}. In our model, in the language of loss networks, the control is on the change of capacity requirements  instead of a change of link. 

\subsection{Scaling Regime}
The invariant distribution being, in general, not known, a scaling approach is used. The network is investigated under Kelly's regime, i.e. under heavy traffic regime with a scaling factor $N$. It has been  introduced in Kelly~\cite{Kelly:2} to study the equilibrium of uncontrolled networks. The arrival rates are scaled by $N$: $\lambda_j$ is replaced by $\lambda_j N$ as well as the capacity  $C$ by $C^N$ and the threshold  $C_0$ by  $C_0^N$ which are such that
\begin{equation}\label{Scale}
C^N=cN+o\left(N\right) \text{ and } C_0^N=c_0N+o\left(N\right),
\end{equation}
for $0{<}c_0{<}c$. 
\begin{definition} For $1{\leq} j{\leq} J$ and $t{\geq} 0$, $L_j^N(t)$ denotes the number of class $j$ jobs at time $t$ in this system and $L^N(t){=}(L_j^N(t),1{\leq} j{\leq} J)$. 
\end{definition}

It will be assumed that the system is overloaded when the jobs have their initial bandwidth requirements
\begin{equation}\tag{$R$}\label{Cond}
\croc{A,\rho} > c \text{ and } \frac{\Lambda}{\mu_1}< c,
\end{equation}
with $\Lambda{=}\lambda_1{+}\cdots{+}\lambda_J$ and $\rho_j{=}\lambda_j/\mu_j$, $1{\leq}j{\leq}J$.
The first condition gives that, without any change on the bandwidth requirement of jobs, the system will reject jobs. The second condition implies that the  network could accommodate all jobs without losses (with high probability)  if all of them would require the reduced bit rate $A_1{=}1$ and service rate $\mu_1$. 

It should be noted that, from the point of view of the design of algorithms, the constant $c_0$ has to be defined.  If one takes $c_0{\in}({\Lambda}/{\mu_1},c)$
then,
\begin{align}
\croc{\rho,A}> c_0,\tag{$R_1$}\label{OVL}\\
\frac{\Lambda}{\mu_1}< c_0. \tag{$R_2$}\label{UVL}
\end{align}
hold.

If $\croc{A,\rho}{<}c$, it is not difficult to see that the system is equivalent to a classical underloaded loss network with one link and multiple classes of jobs. There is, of course, no need to use downgrading policies since the system can accommodate incoming requests without any loss when $N$ is large. See Kelly~\cite{Kelly} or Section~7 of Chapter~6 of Robert~\cite{Robert} for example.

\section{Scaling Results}\label{SecMod}
In this section, we prove convergence results when the scaling parameter $N$ goes to infinity. These results are obtained by studying the asymptotic behavior of the occupation of the link around $C_0^N$,
\begin{equation}\label{mocc}
m^N(t)=\croc{A,L^N(t)}-C_0^N.
\end{equation}
In the context of loss networks, the analogue of such quantity is the number of empty places. 
The following proposition shows that, for the downgrading policy, the boundary $C^N$ does not play a role after some time if Condition~\eqref{UVL} holds. 
\begin{proposition}\label{propbound}
Under Condition~\eqref{UVL} and if the initial state is such that
\[
\lim_{N\to+\infty} \left(\frac{L_j^N(0)}{N}\right) = \ell(0)=\left(\ell_{j,0}\right)\in {\cal S}\stackrel{\text{\rm def.}}{=}\{x\in\R_+^J:\croc{A,x}<c\},
\]
 then, for  $\eps{>}0$, there exists  $t_\eps{\geq}0$ such that, for $T{>}t_\eps$, 
\[
\lim_{N\to+\infty} \P\left(\sup_{t_\eps\leq t\leq T}  \croc{A,L^N(t)}<(c_0{+}\eps)N\right) =1.
\]
\end{proposition}
\begin{proof}
Define
\[
\left(\widetilde{L}_j^N(t)\right)\stackrel{\text{def.}}{=}\left(D_1^N(t){+}X^N(t), D_2^N(t),\ldots, D_J^N(t)\right),
\]
where $(X^N(t))$ is the process   of the number of jobs of an independent  $M/M/\infty$ queue  with $X^N(0){=}0$, service rate  $\mu_1$ and arrival rate $\Lambda{=}\lambda_1{+}\cdots{+}\lambda_J$ and, 
for $1{\leq} j{\leq} J$,
\[
D_j^N(t)=\sum_{k=1}^{L_j^N(0)} \ind{E_{\mu_j,k}>t},
\]
where $(E_{\mu_j,k})$ is a sequence of i.i.d. exponentially distributed random variables with rate $\mu_j$.The quantity $D_j^N(t)$ is the number of initial class $j$ jobs still present at time $t$.   Using Theorem~6.13 of Robert~\cite{Robert}, one gets the convergence in distribution
\[
\lim_{N\to+\infty} \left(\frac{X^N(t)}{N}\right)=\frac{\Lambda}{\mu_1}\left(1-e^{-\mu_1 t}\right),
\]
and, consequently,
\begin{equation}\label{eqa1}
\lim_{N\to+\infty} \left(\frac{1}{N}\croc{A,\widetilde{L}^N(t)}\right)= \left(\frac{\Lambda}{\mu_1}\left(1-e^{-\mu_1 t}\right) +\sum_{j=1}^J A_j\ell_{j,0}e^{-\mu_j t}\right). 
\end{equation}
Since $\mu_1{\leq}\mu_j$ for  $1{\leq}j{\leq}J$,
\begin{multline*}
\frac{\Lambda}{\mu_1}\left(1-e^{-\mu_1 t}\right) +\sum_{j=1}^J A_j\ell_{j,0}e^{-\mu_j t}\\\leq
\frac{\Lambda}{\mu_1}\left(1-e^{-\mu_1 t}\right) +e^{-\mu_1 t}\croc{A,\ell(0)}\leq \max(c_0,\croc{A,\ell(0)}),
\end{multline*}
by Condition~\eqref{UVL}.  Note that the asymptotic occupancy, when $N$ is large, remains below the initial occupancy.

If $0{<}\eps N {<}C^N{-}C^N_0$ and $L^N(0){\in}\N^J$ such that  $ C^N_0{+}\eps N{<} \croc{A,L^N(0)} {<}C^N$, let
\[
\tau_N=\inf\bigg\{t>0: \croc{A,L^N(t)}\leq C_0^N+\frac{\eps}{2}N\bigg\},
\]
 then, on the event $\{\tau_N{>}T\}$, the downgrading policy gives that  the identity in distribution
\begin{equation}\label{eqa3}
\left(\left(\rule{0mm}{4mm}L_j^N(t)\right),0\leq t\leq T\right)\stackrel{\text{dist.}}{=}\left(\left(\rule{0mm}{4mm}\widetilde{L}_j^N(t)\right),0\leq t\leq T\right)
\end{equation}
holds. Condition~\eqref{UVL} gives the existence of $t_\eps$ such that
\[
\frac{\Lambda}{\mu_1}\left(1-e^{-\mu_1 t_\eps}\right) +\sum_{j=1}^J A_j\ell_{j,0}e^{-\mu_j t_\eps}= c_0+\frac{\eps}{2}.
\]
Convergence~\eqref{eqa1} shows that the sequence $(\tau_N)$ converges in distribution to $t_\eps$. 

Note that,  if $S{\in}(t_\eps,T)$,  as long as the process $(\croc{A,L^N(t)})$ stays above $C_0^N$ on   $I{=}[t_\eps,S)$,  a relation similar to~\eqref{eqa3} holds.  By using again Convergence~\eqref{eqa1}, one gets that,  as $N$ goes to infinity,  the process $(\croc{A,L^N(t)}/N)$ remains below $c_0{+}\eps$ with probability close to $1$ on $I$. The proposition is proved.
\end{proof}
We are now investigating the asymptotic behavior of the process $(m^N(t))$ defined by Relation~\eqref{mocc}. The variable indicates if the network is operating in saturation at time $t$, $m^N(t)\geq 0$, or not, $m^N(t)< 0$. 
In pure loss networks, when $N$ is large, up to a change of time scale, the analogue of this process,  the process of the number of empty  places  converges to a reflected random  walk in $\N$. In our case, as it will be seen, the corresponding process is in fact  a random walk on $\Z$. 
\begin{definition}\label{def1}
For $\ell{=}(\ell_j){\in}{\cal S}$, let $(m_\ell(t))$ be the Markov process on $\Z$ whose $Q$-matrix $Q_\ell$  is defined by, for $x{\in} \Z$ and $1{\leq} j{\leq} J$,
\begin{equation}\label{Qmat}
\begin{cases}
Q_\ell(x,x-A_j)=\mu_j\ell_j,\\
Q_\ell(x,x+A_j)=\lambda_j, \text{ if } x< 0,\\
\displaystyle Q_\ell(x,x+1)=\Lambda, \text{ if } x\geq 0,
\end{cases}
\end{equation}
with $\Lambda\stackrel{\text{\rm def.}}{=}\lambda_1{+}\lambda_2{+}\cdots{+}\lambda_J$. 
\end{definition}

The following proposition summarizes the stability properties of the Markov process $(m_{\ell}(t))$. 
\begin{proposition}\label{propm}
If $\ell{=}(\ell_j){\in}{\cal S}$,  then the  Markov process $(m_\ell(t))$ is ergodic if $\ell{\in}\Delta_0$ with 
\begin{multline}\label{E1}
\Delta_0{\stackrel{\text{\rm def.}}{=}}\bigg\{x{\in}{\cal S}{:}\croc{A,x}{=}c_0,\sum_{j=1}^J (\lambda_j{-}\mu_jx_j)A_j{>}0 \text{ and }  \Lambda{<}\sum_{j=1}^J\mu_jx_jA_j\bigg\}
\end{multline}
$\pi_\ell$ denotes the corresponding invariant distribution. 
\end{proposition}
\begin{proof}
The Markov process $(m_\ell(t))$ on $\Z$ behaves like  a random walk on each of the two half-lines $\N$ and $\Z_-^*$. Definition~\eqref{E1}  implies that if $\ell\in\Delta_0$, then  the drift of the random walk  is positive when in  $\Z_-^*$  and negative when in $\N$. This property implies the ergodicity of the Markov process by using the Lyapounov function $F(x){=}|x|$, for example.  See Corollary~8.7 of Robert~\cite{Robert} for example. 
\end{proof}
One now extends the expression $\pi_\ell$ for the values $\ell{\in}{\cal S}\setminus\Delta_0$. This will be helpful to describe the asymptotic dynamic of the system. See Theorem~\ref{thlds} further. 
\begin{definition}\label{Def2}
One denotes  $\pi_\ell{=}\delta_{-\infty}$,  the Dirac measure at ${-}\infty$ when $\ell \in\Delta_-$, with
\[
\Delta_-{\stackrel{\text{\rm def.}}{=}}\bigg\{x{\in}{\cal S}{:} \croc{A,x}{=}c_0, \sum_{j=1}^J (\lambda_j{-}\mu_jx_j)A_j\leq 0\bigg\}\cup\bigg\{x{\in} {\cal S}{:} \croc{A,x}{<}c_0\bigg\},
\]
and $\pi_\ell{=}\delta_{+\infty}$ if $\ell{\in}\Delta_+$, with
\[
\Delta_+{\stackrel{\text{\rm def.}}{=}}\bigg\{x{\in} {\cal S}{:} \croc{A,x}{=}c_0, \sum_{j=1}^J \mu_jx_jA_j\leq \Lambda\bigg\}\cup\bigg\{x{\in} {\cal S}{:} \croc{A,x}{>}c_0\bigg\}.
\]
\end{definition}
\subsection*{Stochastic Evolution Equations}
For $\xi > 0$,  denote by $\cal{N}_{\xi}(\diff t)$ a Poisson process on $\R_+$ with rate $\xi$ and $(\cal{N}_{\xi,i}(\diff t))$  an i.i.d. sequence of such processes. All Poisson processes are assumed to be independent. Classically, the process $(L^N(t))$ can be seen as the unique solution to the following stochastic differential equations (SDE),
\begin{equation}\label{SDE1}
  \begin{cases}
&\displaystyle \diff L^N_1(t)=-\sum_{k=1}^{L^N_1(t-)} {\cal N}_{\mu_1,k}(\diff t)\\& \displaystyle \qquad+\ind{m^N(t-)<C^N-C_0^N}{\cal N}_{\lambda_1N}(\diff t) + \sum_{j=2}^J \ind{0\leq m^N(t-)<C^N-C_0^N}{\cal N}_{\lambda_jN}(\diff t), \\
&\displaystyle \diff L^N_j(t)=-\sum_{k=1}^{L^N_j(t-)} {\cal N}_{\mu_j,k}(\diff t) + \ind{m^N(t-)< 0}{\cal N}_{\lambda_jN}(\diff t),\quad  2\leq j\leq J,
  \end{cases}
\end{equation}
with initial condition $(L^N_j(0))\in\N^J$ such that $\croc{A,L^N(0)}\leq C^N$. 

\begin{theorem}[Limiting Dynamical System]\label{thlds}
Under Condition~\eqref{UVL}, if the initial conditions are such that $m^N(0)=m\in\Z$ and
\[
\lim_{N\to+\infty} \left(\frac{L_j^N(0)}{N}\right)=(\ell_j(0))\in{\cal S},
\]
then there exists continuous process $(\ell(t))= (\ell_j(t))$ such that  the convergence  in distribution
\begin{equation}\label{DS}
\lim_{N\to+\infty} \left(\left(\frac{L_j^N(t)}{N}\right),\int_0^t f\left(m^N(u)\right)\diff u\right)
{=}\left((\ell_j(t)), \int_0^t \int_{\Z} f(x)\pi_{\ell(u)}(\diff x)\diff u\right)
\end{equation}
holds for any function $f$ with finite support on $\Z$.  Furthermore, there exists $t_0{>}0$  such that $(\ell(t), t{\geq} t_0)$ satisfies the differential equations
\begin{equation}\label{dynsys}
\begin{cases}
\displaystyle \frac{\diff}{\diff t} \ell_1(t)=-\mu_1 \ell_1(t)+\lambda_1 +\pi_{\ell(t)}(\N)\left(\sum_{k=2}^J \lambda_k\right),\\
\displaystyle \frac{\diff}{\diff t}\ell_j(t)=-\mu_j\ell_j(t)+\lambda_j \pi_{\ell(t)}(\Z_-^*), \quad 2\leq j\leq J,
\end{cases}
\end{equation}
where $\pi_\ell$, for $\ell{\in}{\cal S}$, is the distribution of  Proposition~\ref{def1} and  Definition~\ref{Def2}. 
\end{theorem}
It should be noted that, since the convergence holds for the convergence in distribution of processes, the limit $(\ell(t))$ is a priori a {\em random } process.
\begin{proof}
By using the same method as Hunt and Kurtz~\cite{Hunt}, one gets the analogue of Theorem~3 of this reference. 
Fix $\eps{>}0$ such that $c_0{+}\eps{<}c$, from Proposition~\ref{propbound}, one gets that the existence of $t_0$ such that
\[
\lim_{N\to+\infty} \P\left(\sup_{t_0\leq t\leq T}  \croc{A,L^N(t)}<(c_0+\eps)N\right) =1,
\]
which implies that the boundary condition $m^N(t){<}C^N{-}C_0^N$ in the evolution equations~\eqref{SDE1}   can be removed. Consequently,  only the boundary condition of $(m^N(t))$ at $0$ plays a role which gives Relation~\eqref{dynsys} as in Hunt and Kurtz~\cite{Hunt}. Note that, contrary to the general situation described in this reference,  we have indeed a convergence in distribution because, for any $\ell\in{\cal S}$,  $(m_\ell(t))$ has exactly one invariant distribution (which may be a Dirac mass at infinity)  by Proposition~\ref{propm}. See Conjecture~5 of Hunt and Kurtz~\cite{Hunt}.
\end{proof}
The following proposition gives  a characterization of the equilibrium point of the dynamical system $(\ell(t))$. 
\begin{proposition}[Fixed Point]\label{FPprop}
Under Conditions~\eqref{OVL} and~\eqref{UVL}, there exists a  unique equilibrium point $\ell^*{\in}\Delta_0$  of the process $(\ell_j(t))$ defined by Equation~\eqref{DS} given by 
\begin{equation}\label{FP}
\begin{cases}
\ell_1^*=c_0-\pi^-(\rho_2A_2{+}\cdots{+}\rho_JA_J),\\
\ell_j^*=\rho_j\pi^-,\quad 2\leq j\leq J,
\end{cases}
\end{equation}
where
\begin{equation}\label{pim}
\pi^-\stackrel{\text{\rm def.}}{=}\frac{c_0-\Lambda/\mu_1}{\croc{A,\rho}-\Lambda/\mu_1},
\end{equation}
with $\Lambda{=} \lambda_1{+}\cdots{+} \lambda_J$.  The process $(m_{\ell^*}(t))$ is ergodic  in this case. 
\end{proposition}
\begin{proof}
Assume that there exists    an equilibrium point $\ell^*=(\ell_j^*)$ of $(\ell_j(t))$ defined by Equation~\eqref{DS}, it is also an equilibrium point   of the dynamical system defined by Equation~\eqref{dynsys}, then
\begin{equation}\label{eq}
\begin{cases}
\mu_1\ell_1^*=\lambda_1+(\lambda_2+\cdots+\lambda_J)(1-\pi^-),\\
\mu_j\ell_j^*=\lambda_j\pi^-,\quad 2\leq j\leq J,
\end{cases}
\end{equation}
with $\pi^-=\pi_{\ell^*}\left(\Z_-^*\right)$. One gets
\begin{equation}\label{aux1}
\sum_{j=1}^J \lambda_j=\sum_{j=1}^J \mu_j\ell_j^*  <\sum_{j=1}^J \mu_j\ell_j^* A_j=\pi^-\sum_{j=1}^J \lambda_j A_j+(1{-}\pi^-)\sum_{j=1}^J \lambda_j <\sum_{j=1}^J \lambda_j A_j.
\end{equation}
We now show that the vector $\ell^*$ is on the boundary, i.e.
\begin{equation}\label{eq1}
\sum_{j=1}^J  A_j\ell_j^*=c_0.
\end{equation}
If we assume that
\[
\lim_{N\to+\infty} \left(\frac{L_j^N(0)}{N}\right)=(\ell_j^*),
\]
from Theorem~\ref{thlds} and the definition of $(m^N(t))$, we know that, for the convergence of processes, the following relation holds
\[
\lim_{N\to+\infty} \left(\frac{m^N(t)}{N}\right)= (\kappa_0), \text{ with } \kappa_0\stackrel{\text{\rm def.}}{=}\sum_{j=1}^J  A_j\ell_j^*-c_0.
\]
For $N_0\in\N$, $\eps>0$ and $N\geq N_0$, 
\[ 
\int_0^1 \ind{|m^N(u)|\geq \eps N}\,\diff u \leq \int_0^1 \ind{|m^N(u)|\geq \eps N_0}\,\diff u.
\]
By using again Theorem~\ref{thlds} and the fact that $\ell^*$ is an equilibrium point of the dynamical system, we have, for the convergence in distribution 
\[
\lim_{N\to+\infty} \int_0^1 \ind{|m^N(u)|\leq \eps N_0}\,\diff u = \pi_{\ell^*}([-\eps N_0,\eps N_0]).
\]
The left-hand side of the above expression can be arbitrarily close to $1$  when $N_0$ is large.  By convergence of the sequence $(m^N(t)/N)$ to $(\kappa_0)$, one gets that, for the convergence in distribution, the relation
\[
\lim_{N\to+\infty} \int_0^1 \ind{|m^N(u)|/N\geq \eps }\,\diff u=0
\]
holds for $\eps{>}0$, which implies that $\kappa_0{=}0$. Thus Relation~\eqref{eq1} holds.  Finally, Relations~\eqref{eq} and~\eqref{eq1} give Relation~\eqref{FP}. 
One concludes therefore that $\ell^*\in\Delta_0$, the associated process $(m_{\ell^*}(t))$ is necessarily ergodic by Proposition~\ref{propm} and Relations~\eqref{aux1}.

To prove that the $\ell^*$ defined by Relations~\eqref{FP} and~\eqref{pim} is indeed an equilibrium point of the dynamical system defined by Equation~\eqref{dynsys},  one has to show that the right-hand side of Equation~\eqref{pim} is indeed    equal to $\pi_{\ell^*}(\Z_-^*)$.   This is  proved in Proposition~\ref{propFP2} of  Section~\ref{SecInv}.  
\end{proof}

\subsection*{Convergence of Invariant Distributions}
In this section our main result establishes the convergence of the invariant distribution of the process $(m^N(t))$ as $N$ gets large. This will give in particular the convergence with respect to $N$ of the probability of not downgrading a request at equilibrium. 
\begin{lemma}\label{lem22}
  If the  process $(\widetilde{L}^N_j(t))$ is the process $(L_j^N(t))$ at equilibrium then, for any $\eps{>}0$ and $T{>}0$, 
  \[
\lim_{N\to+\infty} \P\left(\sup_{0\leq t\leq T}\sup_{2\leq j\leq J}\frac{\widetilde{L}_j^N(t)}{N} \leq \rho_j+\eps \right)=1. 
  \]
\end{lemma}
\begin{proof}
  Let $(L_j^N(t))$ be the process with initial state empty, then one can easily  construct a coupling such that the relation
  \[
  L_j^N(t)\leq\widetilde{Q}_j^N(t), \quad  t\geq 0, \quad 2\leq j\leq J, 
  \]
  holds almost surely, where $(Q_j^N(t))$ is the $M/M/\infty$ queue associated to class $j$ requests. One deduces that,  
  \[
 \widetilde{L}_j^N(0)\leq_{\text{st}} \widetilde{Q}_j^N(0)
 \]
 where $\widetilde{Q}_j^N(0)$ is a Poisson random variable with parameter $\rho_j N$ and $\leq_{\text{st}}$ is the stochastic ordering of random variables. One can therefore construct another coupling such that
   \[
\widetilde{L}_j^N(t)\leq \widetilde{Q}_j^N(t), \quad  t\geq 0, \quad 2\leq j\leq J, 
\]
where $(\widetilde{Q}_j^N(t))$ is a stationary version of the $M/M/\infty$ queue associated to class $j$ requests. The lemma is then a consequence of  the following convergence in distribution of processes, 
\[
\lim_{N\to+\infty} \left(\frac{\widetilde{Q}_j^N(t)}{N}\right)= (\rho_j)
\]
 for $2{\leq} j{\leq} J$, see Theorem~6.13 pp.~159 of Robert~\cite{Robert} for example. 
\end{proof}
\begin{definition}\label{defLA}
  Let  $(y(t))$ be the dynamical system  on ${\cal S}$ satisfying
\begin{equation}\label{dyny}
\begin{cases}
\displaystyle  \frac{\diff}{\diff t}{y}_1(t)=-\mu_1 y_1(t)+\lambda_1 +\left(\sum_{k=2}^J \lambda_k\right)\frac{1}{\Lambda_A}\sum_{k = 1}^{J} A_k\left(\lambda_k- \mu_ky_k(t)\right),\\
\displaystyle  \frac{\diff}{\diff t}{y}_j(t)=-\mu_jy_j(t)+\lambda_j \frac{1}{\Lambda_A}\sum_{k = 1}^{J} \left(A_k \mu_ky_k(t)-\lambda_k\right) , \quad 2{\leq} j{\leq} J,
\end{cases}
\end{equation}
with
\[
\Lambda_A=\sum_{k=1}^J \lambda_k(A_k{-}1). 
\]
\end{definition}

\begin{lemma}\label{lem3}
  If $y(0){\in}\Delta_0$ and if there exists an instant $T{>}0$ such that $y(t){\in}\Delta_0$ for $t{\in}[0,T]$ then
  $(y(t))$ and $(\ell(t)$ coincide on the time interval $[0,T]$,  where $(\ell(t))$ is the solution of Equations~\eqref{dynsys} with $\ell(0){=}y(0)$.
\end{lemma}
\begin{proof}
  The proposition is a simple consequence  of the representation~\eqref{dynsys} of the differential  equations defining the dynamical system $(\ell(t))$  and of the explicit expression of the quantity $\pi_\ell(\Z_-^*)$ given by Relation~\eqref{pil} when $\ell\in\Delta_0$, see Relation~\eqref{E1}.
  \end{proof}
The next proposition investigates the stability Properties of $(y(t))$.
\begin{proposition}\label{expds}
Let $H_0$ be the hyperplane
  \[
  H_0{\stackrel{\text{\rm def.}}{=}}\left\{z\in{\cal S}{:} \croc{A,z}{=}c_0\}\right\}
  \]
  if $y(0){\in}H_0$ then $y(t){\in} H_0$ for all $t{\geq} 0$
and  $(y(t))$ is converging exponentially fast to $\ell^*$ defined in Proposition~\ref{FPprop}.
\end{proposition}
\begin{proof}
  It is easily checked that
  \[ \frac{\diff}{\diff t}\croc{A,  y(t)}=0,  \]
so that if $y(0)\in H_0$, then  the function  $t\mapsto\croc{A,y(t)}$ is constant and equal to $c_0$, hence $y(t)\in H_0$ for all $t\geq 0$.

  For $2{\leq} j{\leq} J$,
\[
\frac{\diff}{\diff t} y_j(t)=\lambda_j b_0 -\mu_jy_j(t)+\lambda_j\sum_{k = 2}^{J} b_ky_k(t),
\]
with
\[
b_0=\frac{\mu_1c_0{-}\Lambda}{\Lambda_A}  \text{ and } b_j=\frac{A_j(\mu_j{-}\mu_1)}{\Lambda_A}.
\]
In matrix form, if $z(t)=(y_2(t),\ldots,y_J(t))$, it can be expressed as 
\begin{equation}\label{aux2}
\frac{\diff}{\diff t}z(t)=  e_b +Bz(t),
\end{equation}
with $e_b=b_0(\lambda_2,\ldots,\lambda_J)\in\R^{J-1}$ and $B=(B_{jk}, 2{\leq} j,k{\leq} J)$ with
\[
B_{jk}=\lambda_j b_k-\mu_j\ind{k=j}.
\]
If $v=(v_2,\ldots,v_J)$ is an eigenvector for the eigenvalue $x$ of $B$, then
\[
(x{+}\mu_j)v_j=\lambda_j\sum_{k = 2}^{J} b_kv_k,\quad  2\leq j\leq J,
\]
hence, $x$ is an eigenvalue if and only if it is a solution of the equation
\[
F(x)\stackrel{\text{\rm def.}}{=}\sum_{j=2}^J \frac{b_j\lambda_j}{x{+}\mu_j}=1.
\]
If $L$ is the number of distinct values of $\mu_j$, $2{\leq}j{\leq}J$, such that $\mu_j{\not=}\mu_1$, then the above equation shows that an eigenvalue is a zero of a polynomial of degree at most $L$. Using  Conditions~\eqref{Cond}, it is easy to check that  the relation $F(0){<}1$ holds. In particular $0$ is not an eigenvalue and, consequently $B$ is  invertible. Due to the poles of $F$ at the $-\mu_j$, $2{\leq}j{\leq}J$ and the relations  $F(0){<}1$ and $\mu_j{\geq} \mu_1$ for $2{\leq}j{\leq}J$, one has already $L$ negative solutions of the  equation $F(x){=}1$. All eigenvalues of $B$ are thus negative, consequently, $\exp(tB)$ converges to $0$. (See Corollary~2 of Chapter~25 of Arnol'd~\cite{Arnold} for example.)

Equation~\eqref{aux2} can be solved as
\[
z(t)=e^{t B}\left(z(0){+}B^{{-}1} e_b\right)-B^{-1} e_b.
\]
Therefore the function $(z(t))$ has a limit at infinity given by ${-}B^{-1} e_b$ which is clearly $(\ell^*_j,2{\leq} j{\leq} J)$. The proposition is proved.
\end{proof}
One can now prove the main result of this section. 
\begin{theorem}\label{theoeq}
If $\ell^*$ is  the quantity defined in Proposition~\ref{FPprop}, then  the equilibrium distribution of $(m^N(t))$ converges to $\pi_{\ell^*}$ when $N$ goes to infinity.
\end{theorem}
\begin{proof}
Recall that $m^N(t){=}\croc{A,L^N(t)}{-}C_0^N$ and let $\Pi^N$ be the invariant distribution of $(L^N(t))$. It is  assumed that the distribution of $L^N(0)$ is $\Pi^N$ for the rest of the proof. 
In particular $(m^N(t))$ is a stationary process. 

One first proves that $(L^N(0)/N)$ converges in distribution to $\ell^*$. The boundary condition $\croc{A,L^N(0)}{\leq} C^N$ gives that the sequence of random variables $(L^N(0)/N)$ is tight.  If $(L^{N_k}(0)/N_k)$ is a convergent subsequence to some random variable $\ell^\infty$, by Theorem~\ref{thlds}, one gets that, for the convergence in distribution, the relation
\[
\lim_{k\to+\infty} \left(\left(\frac{L^{N_k}(t)}{N_k}\right)\right)
=\left(\ell(t)\right)
\]
holds, where $(\ell(t))$ is a solution of Equation~\eqref{dynsys} with initial point at $\ell(0)=\ell^\infty$. Note that $(\ell(t))$ is a stationary process, its distribution is invariant under any time shift. 

By Lemma~\ref{lem22} one has that the relation  $\ell_j(t){\leq} \rho_j$, for $2{\leq} j{\leq} J$, holds almost surely on any finite time interval and, by Proposition~\ref{propbound},  $\croc{A,\ell(t)}{\leq} c_0$ also holds almost surely on finite time intervals.

Assume that $\croc{A,\ell(0)}{<}c_0$ holds. The ODEs defining the limiting dynamical system are given by
\[
\displaystyle  \frac{\diff}{\diff t} \ell_j(t)=-\mu_j\ell_j(t)+\lambda_j, \quad 1\leq j\leq J,
\]
as long as the condition $\croc{A,\ell(t)}{<}c_0$ holds, hence on the corresponding time interval, one has 
\[
  \ell_j(t)=\rho_j+(\ell_j(0)-\rho_j)e^{-\mu_j t},\quad 1 \leq j\leq J,
  \]
so that 
  \[
 \croc{A,\ell(t)}= \croc{A,\rho}+\sum_{j=1}^J A_j\left(\ell_j(0)-\rho_j\right)e^{-\mu_j t}.
 \]
 Since $\croc{A,\rho}{>}c_0$,  there exists some $t_1{>}0$ such that $\croc{A,\ell(t_1)}{=}c_0$.

 Hence, by stationarity in distribution of $(\ell(t))$, one can shift time at $t_0$ and assume that $\croc{A,\ell(0)}{=}c_0$.  On this event
\begin{equation}\label{eqq1}
\sum_{j=1}^J \mu_j\ell_j(0)A_j\geq \mu_1\sum_{j=1}^J\ell_j(0)A_j=\mu_1 c_0>\Lambda=\sum_{j=1}^J\lambda_j. 
\end{equation}
Similarly, since $\ell_j(0){\leq} \rho_j$ for all $2{\leq} j{\leq} J$, 
\begin{align}
&\sum_{j=1}^J A_j(\lambda_j{-}\mu_j\ell_j(0)){=}\lambda_1{-}\mu_1c_0{+}\mu_1\sum_{j=2}^J A_j\ell_j(0){+}
\sum_{j=2}^J A_j(\lambda_j{-}\mu_j\ell_j(0)) \label{eqq2} \\
&=-\mu_1c_0{+} \sum_{j=1}^J A_j(\lambda_j{+}(\mu_1{-}\mu_j)\ell_j(0))
\geq {-}\mu_1c_0+ \sum_{j=1}^J A_j(\lambda_j+(\mu_1{-}\mu_j)\rho_j)\notag\\
&= -\mu_1c_0+ \sum_{j=1}^J A_j\lambda_j\frac{\mu_1}{\mu_j}= \mu_1\left(\croc{A,\rho}-c_0\right)>0,\notag
\end{align}
and the last quantity is independent of $\ell(0)$. Relations~\eqref{eqq1} and~\eqref{eqq2} show that $\ell(0){\in}\Delta_0$ and, by Equations~\eqref{dynsys} and \eqref{dyny},  they also hold for $t$ in a small neighborhood $I$ of $0$ independent of $\ell(0)$ so that $\ell(t){\in}\Delta_0$ for $t{\in}I$.  Consequently,  the dynamical system $(\ell(t))$ never leaves $\Delta_0$.  Lemma~\ref{lem3}  shows that the two dynamical systems $(\ell(t))$ and $(y(t))$ (with $y(0){=}\ell(0)$) coincide. Hence, on one hand $(\ell(t))$ is a stationary process and, on the other hand, it is a dynamical system  converging to $\ell^*$\!, one deduces that it is constant and equal to $\ell^*$. We have thus proved that the sequence $(L^N(0)/N)$ converges in distribution to $\ell^*$. 

Using again Theorem~\ref{thlds}, one gets that, for the convergence in distribution,
\[
\lim_{N\to+\infty}  \int_0^1 f(m^N(u))\,\diff u
= \int_{\Z} f(x)\pi_{\ell^*}(\diff x)
\]
holds for any function $f$ with finite support on $\Z$. By using the stationarity of $(m^N(t))$ and Lebesgue's Theorem, one obtains
\[
\lim_{N\to+\infty}  \E\left( f(m^N(0))\right) = \int_{\Z} f(x)\pi_{\ell^*}(\diff x).
\]
The theorem is proved. 
\end{proof}
Since a job arriving at time $t$ is not downgraded if $m^N(t){<}0$, one obtains the following corollary. 
\begin{corollary}\label{corol}
As $N$ goes to infinity, the probability that, at  equilibrium,  a job is not downgraded in this allocation scheme is converging to $\pi^-$ defined in Proposition~\ref{FP},
\[
\pi^-=\frac{c_0{-}\Lambda/\mu_1}{\croc{A,\rho}{-}\Lambda/\mu_1}.
\]
\end{corollary}
\section{Invariant Distribution}\label{SecInv}
We  assume in this section that $\ell{\in}\Delta_0$, as  defined in Proposition~\ref{propm}, so that $(m_\ell(t))$ is an ergodic Markov process.  The goal of this section is to derive an explicit expression of the invariant distribution $\pi_\ell$ on $\Z$ of $(m_\ell(t))$. At the same time, Proposition~\ref{propFP2} below  gives the required argument to complete the proof of  Proposition~\ref{FPprop} on the characterization of the fixed point of the dynamical system. 
\subsection{Functional Equation}
 In the following we denote by $Y_\ell$ a random variable with distribution $\pi_\ell{=}(\pi_\ell(n), n{\in}\Z)$.  

For $r{>}0$, we will use the  notation
\[
D(r){=} \{z \in \C,|z|{<} r\},\quad D^c(r){=}\{z \in \C, |z|{>}r\}\text{ and }\gamma(r){=}\{z\in \C, |z|{=}r\}. 
\]
For sake of simplicity, we will use $D{=}D(1)$ and $D^c{=}D^c(1)$.

\begin{lemma}
With the notation
\[
 \varphi_{+}(z)= \E\left(z^{Y_{\ell}} \ind{{Y_{\ell}} \geq 0}\right), \quad 
 \varphi_{-}(z)= \E\left(z^{Y_{\ell}} \ind{{Y_{\ell}} < 0}\right),
\]
the random variable $Y_\ell$ is such that
\begin{equation}\label{WH}
P_1(z)	\varphi_{+}(z) = P_2(z) \varphi_{-}(z) 
\end{equation}
 where $P_1$ and $P_2$ are  polynomials defined by 
 \begin{equation}\label{Def_Ps}
   \begin{cases}
	\displaystyle	P_1(z)  = \sum\limits_{j=1}^J\left[\rule{0mm}{4mm}
			(\lambda_j+ \mu_j \ell_j) z^{A_J} - \lambda_j z^{{A_J}+1} - \mu_j \ell_j z^{{A_J}-A_j} \right],\\
	\displaystyle	P_2(z)  = \sum\limits_{j=1}^J \left[\rule{0mm}{4mm}
			 \lambda_j z^{{A_J} + A_j} + \mu_j \ell_j z^{{A_J}-A_j} -(\lambda_j + \mu_j \ell_j) z^{A_J}	\right].
   \end{cases}
 \end{equation}
\end{lemma}

\begin{proof}
 For  $z \in  \gamma(1)$ define $f_z:\Z\mapsto \C$ such that $f_z(x){=}z^x$, for $x\in\Z$.  Equilibrium  equations for $(m_\ell(t))$ give the identity
\[
\sum_{\substack{x,y\in\Z\\x\not=y}} \pi_\ell(x)Q_\ell(x,y)(f_z(y)-f_z(x))=0,
\]
where $Q_\ell$ is the $Q$-matrix of $(m_\ell(t))$ given by Equation~\eqref{Qmat}.  After some simple reordering, one gets the relation
\begin{multline}\label{Gener Func}
	\E\left( z^{Y_{\ell}} \ind{{Y_{l}} \geq 0} \right)
    \sum_{j=1}^J \left(\lambda_j (1 - z) + \mu_j \ell_j \left(1 - z^{-A_j}\right) \right)= \\
	- \E\left(z^{Y_{l}} \ind{{Y_{l}} < 0} \right)
	\sum_{j =1}^J \left(\lambda_j\left(1 - z^{A_j}\right) + \mu_j \ell_j \left(1 - z^{-A_j}\right) \right).
\end{multline}
By using the definition of $ \varphi_{+}(z)$ and $ \varphi_{-}(z)$, Equation~\eqref{Gener Func} can be rewritten as Equation~\eqref{WH}.
\end{proof}

\begin{proposition} \label{propFP2}
  If $\ell\in\Delta_0$ then
\begin{equation}\label{pil}
\pi_{\ell}(\Z_{-}^*) =  \frac{\sum_{j = 1}^{J} \left(A_j \mu_j \ell_j {-} \lambda_j \right)}{\sum_{j=1}^{J}\lambda_j (A_j{-}1)}.
\end{equation}
In particular if $\ell^*{\in}{\cal S}$ is given by Relation~\eqref{FP} then
\[
\pi_{\ell^*}(\Z_{-}^*)=\frac{c_0{-}\Lambda/\mu_1}{\croc{A,\rho}{-}\Lambda/\mu_1}.
\]
\end{proposition}
Note that the right-hand side of the last relation is precisely $\pi^-$ of Relation~\eqref{pim} which is the result necessary to complete the proof of  Proposition~\ref{FPprop}. 
\begin{proof} 
With the same notations as before, from Relation~\eqref{WH}, 
\[
\frac{\varphi_{-}(z)}{\varphi_{+}(z)} = \frac{P_1(z)}{P_2(z)}
\]
holds for $z{\in}\C$, with $z{\in}\gamma(1)$. By definition of $\varphi_{-}(z)$ and $\varphi_{+}(z)$,
\[
\lim_{z \to 1} \varphi_{-}(z)=\pi_{\ell}(\Z_{-}^*) \text{ and } \lim_{z \to 1} \varphi_{+}(z)=1-\pi_{\ell}(\Z_{-}^*).
\]
Since $1$ is a zero of $P_1$ and $P_2$, this gives the relation
\[
\frac{\pi_{\ell}(\Z_{-}^*) }{1-\pi_{\ell}(\Z_{-}^*) }=\frac{P_1'(1)}{P_2'(1)}=
\frac{\sum_{j = 1}^{J} \left(A_j \mu_j \ell_j - \lambda_j \right)}{\sum_{j=1}^{J} A_j \left(\lambda_j -\mu_j \ell_j \right)}.
\]
Using the expression of $(\ell_j^*)$, with some algebra, one gets
\[
\pi_{\ell^*}(\Z_{-}^*) = \left(c_0-\sum_{j=1}^{J} \frac{\lambda_j}{\mu_1}\right) \Bigg{/} \left( \sum\limits_{j=1}^{J} \rho_j A_j-\sum\limits_{j=1}^{J} \frac{\lambda_j}{\mu_1} \right)=\pi^-.
\]
The proposition is proved. 
\end{proof}
Relation~\eqref{WH} is valid on the unit circle, however  the function $\varphi_{+}$ (resp. $\varphi_{-}$) is defined on $D$ (resp. $D^c$). This can then be expressed as a Wiener-Hopf factorization problem analogous to the one used in the analysis of reflected random walks on $\N$. This is used in the analysis of the $GI/GI/1$ queue, see Chapter~VIII of Asmussen~\cite{Asmussen} or Chapter~3 of Robert~\cite{Robert} for example. In a functional context, this is a special case of a Riemann's problem, see Gakhov~\cite{Gakhov}. In our case, this is a  random walk in $\Z$, with a  drift depending on the half-space where it is located. The first (resp. second) condition in the definition of the set $\Delta_0$ in Definition~\eqref{E1} implies that the drift of the random walk in $\Z_{-}^*$ (resp. in $\N$) is positive (resp. negative). 

The first step in the analysis of Equation~\eqref{WH} is to  determine the locations of the zeros of $P_1$ and $P_2$. This is the purpose of the following lemma. 
\begin{lemma}\label{Lemma1}\label{Lemma2}(Location of the Zeros of $P_1$ and $P_2$) Let  $\ell$ be in $\Delta_0$.
\begin{enumerate}
\item[{\em (i)}] Polynomial $P_2$ has exactly two positive real roots $1$  and  $z_2 {\in}]0,1[$. There are ${A_J}{-}1$ roots   in $D(z_2)$ and ${A_J}{-}1$  roots whose  modulus are strictly greater than~$1$.
\item[{\em (ii)}] Polynomial $P_1$ has exactly two positive real roots $1$  and  $z_1{>}1$. The ${A_J}{-}1$ remaining roots have a modulus strictly smaller than~$1$.
\end{enumerate}
\end{lemma}
\begin{proof}
One first notes that  $P_2$ is a  polynomial with the same form as the $f$ defined by Equation~(13) in Bean et al.~\cite{Bean} (with  $e_j {=} A_j$, $\kappa_j {=} \lambda_j$ and $\hat{e} {=} {A_J}$).  The roots of $Q$ are exactly the roots of $f$.  Lemma~2.2 of Bean et al.~\cite{Bean} gives assertion~({\em i}) of our lemma. 

The proof of assertion~{\em (ii)} uses an adaptation of the argument for the proof of  Lemma~2.2 of Bean et al.~\cite{Bean}.  Define the function  $f(z) {=} z^{-A_J}P_1(z)$. Recall that $P_1$  is a polynomial with degree ${A_J}{+}1$. 
There are exactly two real positive roots for $P_1$.  Indeed, $f(1) = 0$ and it is easily checked that $f$ is strictly concave with
\[
f'(1)= \sum\limits_{j=1}^J (-\lambda_j + A_j \mu_j \ell_j) > 0,
\]
since $\ell{\in}\Delta_0$, by the second condition in Definition~\eqref{E1}.  Hence $P_1$ has a real zero $z_1$ greater than $1$. 

Let $r \in (1,z_1)$ be fixed, note that  $P_1(r) {>} 0$. Define
\begin{align*}
\displaystyle f_1(z) &= K z^{A_J},\text{ with } K= \sum_{j=1}^J \left(\lambda_j + A_j \mu_j \ell_j \right), \\
\displaystyle f_2(z) &= \sum_{j=1}^J \left(\lambda_j z^{{A_J}+1} + \mu_j \ell_j z^{{A_J}-A_j}\right),
\end{align*}
so that $P_1{=}f_1{-}f_2$.  

Fix some  $z\in \gamma(r)$. By  expressing these functions in terms of real and imaginary parts, 
\[
z^{A_J} = \alpha_1+ i \beta_1 \text{ and }f_2(z) = \alpha_2 + i \beta_2, 
\]
one gets 
\begin{multline}\label{eqb}
        \left|\rule{0mm}{4mm}f_1(z) - f_2(z) - b z^{A_J}\right|^2 = |K(\alpha_1+  i \beta_1) - b(\alpha_1+ i \beta_1)-(\alpha_2+ i \beta_2)|^2\\
                                                                        = (K\alpha_1-\alpha_2)^2+(K\beta_1-\beta_2)^2+H= \left|f_1(z)-f_2(z)\right|^2+H,
\end{multline}
with 
\begin{multline*}
        H =(b\alpha_1)^2 - 2b\alpha_1(K\alpha_1-\alpha_2)+(b\beta_1)^2-2b\beta_1(K\beta_1 - \beta_2)\\
          = b(b- 2 K)(\alpha_1^2+\beta_1^2) + 2b(\alpha_1\alpha_2 + \beta_1\beta_2).
\end{multline*}
Cauchy-Schwarz's Inequality gives the relation
\[
\alpha_1\alpha_2+\beta_1\beta_2\leq \frac{1}{K} |f_2(z)||f_1(z)| \leq \frac{1}{K} f_2(r)f_1(r),
\]
since $|f_i(z)|{\leq} f_i(|z|)$ for $i{=}1$, $2$. Thus, 
\begin{multline*}
        \frac{H}{b}  = (b - 2K)(\alpha_1^2+\beta_1^2)+2(\alpha_1\alpha_2+\beta_1\beta_2)
                                \leq (b - 2K) \frac{f_1(r)^2}{K^2} + 2 f_2(r) \frac{f_1(r)}{K}\\
                                = \frac{f_1(r)}{K^2}\bigg((b - 2K)f_1(r)+2 K f_2(r)\bigg)
                                = \frac{f_1(r)}{K^2}(b f_1(r) - 2 K P_1(r)).
\end{multline*}
Since $P_1(r){>}0$,  $b$ can be chosen so that $b f_1(r) {<}2 K P_1(r)$. From the above relation and  Equation~\eqref{eqb},  one gets that for $z{\in}\gamma(r)$, the relation
\[
\left|f_1(z)-f_2(z)-b z^{A_J}\right|< |f_1(z)-f_2(z)|
\]
holds. By Rouch\'e's theorem, one obtains that, for any $r{\in}(1,z_1)$,  $P_1$ has exactly ${A_J}$ roots in $D(r)$. One concludes that $P_1$ has exactly ${A_J}$ roots in $\overline{D}$.  It is easily checked that if $z{\in}\gamma(1)$ and $z{\not\in}\R$ then the real part of $P_1(z)$ is positive, hence $z$ cannot be a root of the polynomial $P_1$. Consequently, $P_1$ has exactly ${A_J}{-}1$ roots in $D$. The lemma is proved. 
\end{proof}

\begin{definition}\label{def2}
For $U\in\{P_1, P_2\}$, denote by ${\cal Z}_U$ the set of the zeros  of $U$ different from $1$.
\end{definition}
 Define
\begin{equation*} 
\Phi(z)=
\begin{cases}
\displaystyle -\varphi_{+}(z)  {\lambda_J}^{-1} (z-z_1)\prod_{q\in {\cal Z}_{P_2}\cap D^c}{(z-q)^{-1}},&z\in D \\\ \\
\displaystyle  \varphi_{-}(z) {\Lambda}^{-1} \prod_{q\in {\cal Z}_{P_2}\cap D}(z-q)\prod_{p\in{\cal Z}_{P_1}\cap D}(z-p)^{-1},&z\in D^c
\end{cases}
\end{equation*}
with $\Lambda{=}\lambda_1{+}\cdots{+}\lambda_J$  and the same notations as before.  By definition, function $\Phi$ is holomorphic in $D$ and $D^c$ and, from Relation~\eqref{WH}, is continuous on $\gamma(1)$.  The analytic continuation theorem, Theorem~16.8 of Rudin~\cite{Rudin} for example, gives that $\Phi$ is holomorphic on $\C$.  For $z\in D^c$, 
\[
|\varphi_{-}(z)|\leq \E\bigg(\ind{Y_{\ell}<0}|z|^{Y_{\ell}}\bigg) \leq \frac{1}{|z|},
\]
since the  cardinality of ${\cal Z}_{P_1}\cap D$ (resp. ${\cal Z}_{P_2}\cap D$) is $A_J{-}1$ (resp. $A_J$), the holomorphic function $\Phi$ is therefore bounded on $\C$. 
By Liouville's theorem,   $\Phi$ is constant, equal to $\kappa  \in \C$. Therefore
\begin{equation}\label{GenFunc}
\begin{cases}
\displaystyle	\varphi_{+}(z)\displaystyle{=}  -\kappa \lambda_J (z-z_1)^{-1}\prod_{q\in {\cal Z}_{P_2}\cap D^c}(z-q), \qquad  &z\in D, \\
\displaystyle	\varphi_{-}(z)\displaystyle{=}\kappa \Lambda \prod_{q\in {\cal Z}_{P_2}\cap D}(z-q)^{-1}\prod_{p\in{\cal Z}_{P_1}\cap D}(z-p),\qquad &z\in D^c.
\end{cases}
\end{equation}
Recall that $\varphi(z)  = \varphi_{+}(z) + \varphi_{-}(z)= \E\left( z^{Y_{\ell}} \right)$ is a generating function,  in particular $ \varphi(1) =1$. Plugging the previous expressions for $\varphi_{+}$ and $\varphi_{-}$ in $\varphi_{+}(1)+\varphi_{-}(1)=1$, one gets the relation
\[
1=-\kappa \prod_{q\in {\cal Z}_{P_2}\cap D}(1-q)^{-1}\frac{1}{1-z_1}  \left( 
P_1'(1)+P_2'(1)\right),
\]
hence, using equation~\eqref{Def_Ps},
\[
\kappa=\frac{z_1{-}1}{\Lambda_A} \prod_{q\in {\cal Z}_{P_2}\cap D}(1{-}q),
\]
where $\Lambda_A$  is introduced in Definition~\ref{defLA}. Note that $\kappa$ is positive. We can now state the main result of this section. 
\begin{proposition}[Invariant Measure]\label{invariant} If $\ell{\in}\Delta_0$ defined by Relation~\eqref{E1}, then the invariant measure $\pi_{\ell}$ can be expressed, for $n{\in} \Z$, as 
\[
 \pi_\ell(n) = 
\begin{cases}
\displaystyle -\kappa \sum_{q\in {\cal Z}_{P_2}\cap D} \frac{P_1(q)q^{-n-1}}{(q-z_1)(q-1)R_D'(q)},  & n{<}0,\\
\displaystyle \kappa \bigg(\alpha_n+\frac{P_2(z_1)z_1^{-n-1}}{(z_1-1)R_D(z_1)}\bigg),  & 0{\leq} n {<}A_J{-}1,\\
\displaystyle \kappa \frac{P_2(z_1)z_1^{-n-1}}{(z_1-1)R_D(z_1)},  & n{\geq} A_J{-}1,
\end{cases}
\]
where $z_1$ is defined in Lemma~\ref{Lemma1},  and  $P_1$ and $P_2$  by Relation~\eqref{Def_Ps}, 
\[
R_D(z)=\prod_{q\in {\cal Z}_{P_2}\cap D} (z-q),\quad \kappa=\frac{(z_1{-}1)R_D(1)}{\Lambda_A},
\]
for $0\leq n<A_J-1$,  $\alpha_n$ is the coefficient of degree $n$ of the polynomial
\[
-\frac{1}{z-z_1}\left(\frac{P_2(z)}{(z-1)R_D(z)}-\frac{P_2(z_1)}{(z_1-1)R_D(z_1)}\right).
\]
\end{proposition}
\begin{proof}
Note that, for $z\in\C$,
\[
\prod_{p\in{\cal Z}_{P_1}\cap D}(z-p)=-\frac{1}{\Lambda}\frac{P_1(z)}{(z-z_1)(z-1)}.
\]
For $z\in D^c$, 
\[
\varphi_{-}(z)=\kappa \Lambda \prod_{q\in {\cal Z}_{P_2}\cap D}(z-q)^{-1} \prod_{p\in{\cal Z}_{P_1}\cap D}(z-p).
\]
Since $|{\cal Z}_{P_1}\cap D|{=}A_J{-}1{<}A_J{=}|{\cal Z}_{P_2}\cap D|$ by Lemma~\ref{Lemma1}, $\varphi_{-}$  has the following partial fraction decomposition
\begin{align*}
\varphi_{-}(z)&=-\kappa  \sum_{q\in {\cal Z}_{P_2}\cap D} \frac{P_1(q)}{(q-z_1)(q-1)R_D'(q)}\frac{1}{z-q}\\
&=\sum_{i=0}^{\infty} -\kappa \sum_{q\in {\cal Z}_{P_2}\cap D} \frac{P_1(q)q^i}{(q-z_1)(q-1)R_D'(q)}  \frac{1}{z^{i+1}}.
\end{align*}
Denote
\[
R_{D^c}(z)=\prod_{q\in {\cal Z}_{P_2}\cap D^c}(z-q)=\frac{P_2(z)}{\lambda_J(z-1)R_D(z)},
\]
then
\[
\varphi_{+}(z)=-\kappa \lambda_J \frac{R_{D^c}(z)}{z-z_1}=\kappa \bigg(-\lambda_J\frac{R_{D^c}(z)-R_{D^c}(z_1)}{z-z_1}+ \frac{P_2(z_1)}{(1-z_1)R_D(z_1)}\frac{1}{z-z_1}\bigg).
\]
One concludes by using the expression of $\kappa$ obtained before. 
\end{proof}

\subsection{Some Moments of $\mathbf{(\pi_{\ell^*})}$} Using the probability generating function $\varphi(z)$ of $\pi_{\ell^{*}}$  from Equation~\eqref{GenFunc},  one can derive an explicit expression of the mean, the variance and the skewness  of such distribution.  The skewness of a random variable $X$ is a measure of the asymmetry of the distribution of $X$,
\[
\mathrm{Skew}(X)\stackrel{\text{def.}}{=}{\E([X-\E(X)]^3)}
\]
See Doane and Seward~\cite{Doane} for example. 
\begin{proposition}\label{mean_variance_skew}
If $Y_{\ell^{*}}$ is a random variable with distribution $\pi_{\ell^{*}}$ then
	\begin{align*}
		\E(Y_{\ell^{*}}) &= A_J + \frac{\theta _2}{2\theta_1} - S(1),\\
		\Var(Y_{\ell^{*}}) &=\frac{\theta _2 + 2\theta_3}{6\theta_1} - \left(\frac{\theta_2}{2\theta_1}\right)^2 - \left( S(1) + S'(1) \right),\\
		\mathrm{Skew}(Y_{\ell^{*}}) &= \frac{\theta_2^3}{4\theta_1^3} + \theta_2 \frac{\theta_2 - 2 \theta_3}{4 \theta_1^2} + \frac{\theta_4 - \theta_3}{4 \theta_1} - \left( S(1) + 3S'(1) + S''(1) \right),
	\end{align*}
where, for $i\geq 1$, 
\[
\theta_i = \sum\limits_{j=2}^J \lambda_j A_j^{i-1}(A_j-1),
\]
 and
\[
S(z) = \dfrac{1}{z-z_1}+\sum\limits_{q\in {\cal Z}_{P_2}\cap D} \dfrac{1}{z-q},
\]
with $R_D(z)$ defined in Proposition~\ref{invariant}. 
\end{proposition}
The proof is straightforward, modulo some tedious calculations of the successive derivatives of $\varphi(z)$ evaluated at $1$.
Figure~\ref{FigDist} shows that the distribution of $Y_{\ell^*}$ is significantly asymmetrical. For this example $\E(Y_{\ell^{*}}){=}8.04819$, $\Var(Y_{\ell^{*}}){=}77.2284$ and $\mathrm{Skew}(Y_{\ell^{*}}) =0.967069$. 
\begin{figure}[ht]
	\includegraphics[scale=0.4]{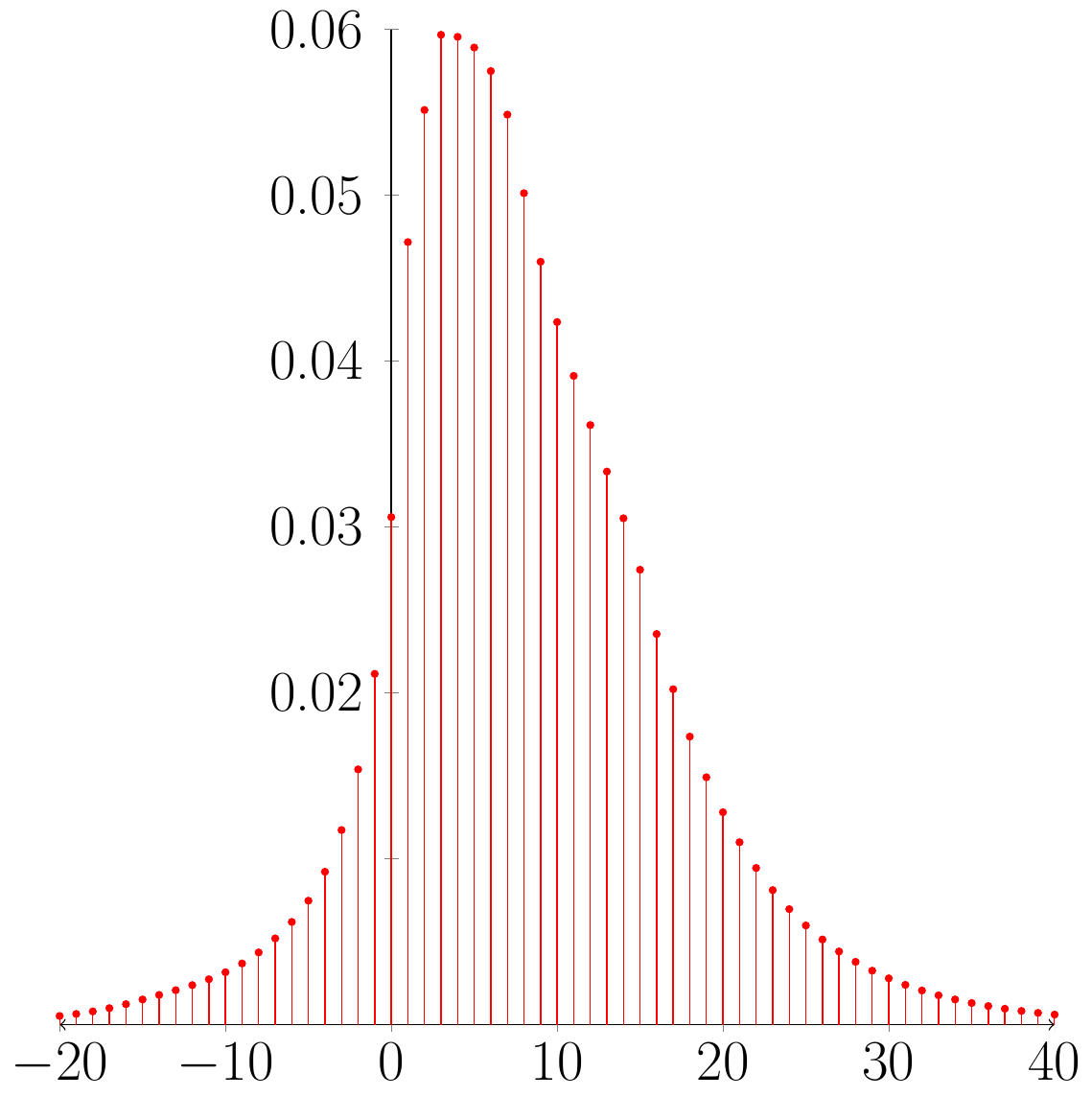}
	\caption{The histogram of $Y_{\ell^*}$ with the parameters $J{=}5$, $A{=}(1, 2, 4, 8, 16)$, $\lambda{=}(0.25,0.2,0.15,0.1,0.05)$, $\mu{=}(1,1,1,1,1)$ and $c_0{=}0.97$.}\label{FigDist}
\end{figure}

\section{Applications}\label{App}
\subsection*{Comparison with a Pure Loss System}
In this case, a request which cannot be accommodated is rejected right away.   Recall that, with probability $1$, our algorithm does not reject any request. The purpose of this section is to discuss the price of such a policy.  Intuitively,  at equilibrium  the probability $W_L$ of accepting a job at requested capacity   in a pure loss system  is greater that the corresponding quantity $W_D$  for the downgrading algorithm.  See Proposition~\ref{propcmp} below. A further question is to assess the impact of such policy, i.e. the order of magnitude of the difference $W_L{-}W_D$. 

Under the same assumptions about the arrivals and under the condition
\begin{equation}\tag{$R$}
\croc{A,\rho}{>}c\quad \text{ and } \quad \frac{\Lambda}{\mu_1}{<}c,
\end{equation}
with $\Lambda=\lambda_1{+}\cdots{+}\lambda_J$,  then, as $N$ gets large, the equilibrium probability that a request of class $1{\leq}j{\leq}J$   is
accepted in the pure loss system  is converging to  $\beta^{A_j}$, 
where $\beta\in(0,1)$ is the unique solution of the equation
\begin{equation}\label{FPL}
\sum_{j=1}^J A_j\rho_j\beta^{A_j}=c. 
\end{equation}
see Kelly~\cite{Kelly:2}.  Consequently, the asymptotic load of accepted requests is given by 
\[
W_L\stackrel{\text{\rm def.}}{=}\frac{1}{\Lambda}\sum_{j=1}^J \rho_j \beta^{A_j}.
\]
Under the downgrading policy, the equilibrium probability that a job 
is accepted without degradation is given by $\pi^-$, the asymptotic load of  requests accepted without degradation is 
\[
W_D\stackrel{\text{\rm def.}}{=}\frac{1}{\Lambda} \sum_{j=1}^J \rho_j \frac{c_0{-}\Lambda/\mu_1}{\croc{\rho,A}{-}\Lambda/\mu_1},
\]
for $c_0\in (\Lambda/\mu_1,c)$. Note that, when the service rates are constant equal to $1$, then $W_L$ (resp. $W_D$) is the asymptotic throughput of accepted requests (resp. of non-degraded requests). 

\medskip
The following proposition establishes the intuitive property that a pure loss system has better performances in terms of acceptance. 
\begin{proposition}\label{propcmp}
For $c_0\in (\Lambda/\mu_1,c)$, the relation  $W_D{\leq} W_L$ holds. 
\end{proposition}
\begin{proof}
The representation of these quantities gives that  the relation to prove is equivalent to the inequality
\[
\sum_{j=1}^J  \rho_j\beta^{A_j} \left(\sum_{j=1}^J \rho_j A_j{-}\sum_{j=1}^J \frac{\lambda_j}{\mu_1} \right)-\sum_{j=1}^J\rho_j\left(c_0{-}\sum_{j=1}^J\frac{\lambda_j}{\mu_1}\right)\geq 0.
\]
By using the fact that $c_0{<}c$ and Equation~\eqref{FPL}, it is enough to show that the quantity
\[
\Delta\stackrel{\text{\rm def.}}{=}\sum_{j=1}^J \rho_j \beta^{A_j} \left(\sum_{i=1}^J\rho_j A_i-\sum_{i=1}^J \frac{\lambda_i}{\mu_1} \right){-}\sum_{j=1}^J\rho_j\left(\sum_{i=1}^J A_i\rho_j\beta^{A_i}{-}\sum_{i=1}^J\frac{\lambda_i}{\mu_1}\right)
\]
is positive. But this is clear since
\begin{multline*}
\Delta=\sum_{1\leq i,j\leq J} \rho_i \rho_j\left(\rule{0mm}{4mm} A_j\left(\beta^{A_i}{-}\beta^{A_j}\right)\right) +\sum_{1\leq i,j\leq J} \rho_j \frac{\lambda_i}{\mu_1}\left(1{-}\beta^{A_j}\right)\\=
\sum_{1\leq i< j\leq J} \rho_i \rho_j\left(\rule{0mm}{4mm}(A_j{-}A_i)\left(\beta^{A_i}{-}\beta^{A_j}\right)\right)
+\sum_{1\leq i,j\leq J} \rho_j \frac{\lambda_i}{\mu_1}\left(1{-}\beta^{A_j}\right)
\end{multline*}
and the terms of both series of the right hand side of this relation are non-negative due to the fact that $0{<}\beta{<}1$. 
\end{proof}
Numerical experiments have been done to estimate the difference $W_L{-}W_D$, see Figure~\ref{FigCCMP}. The general conclusion is that, at moderate load under Condition~($R$),  the downgrading algorithm performs quite well with only a small fraction of downgraded jobs. As it can be seen this is not anymore true for high load where, as expected, most of requests are downgraded but nobody is lost. 
\begin{figure}
        \centering
                \subfigure[{$J{=}2,  A_2{=}3,\lambda_1{=}0.2$,}\newline{$(R)$ conditions{:} $\lambda_2{\in} (0.2633, 0.79)$}]{\includegraphics[width=.45\textwidth]{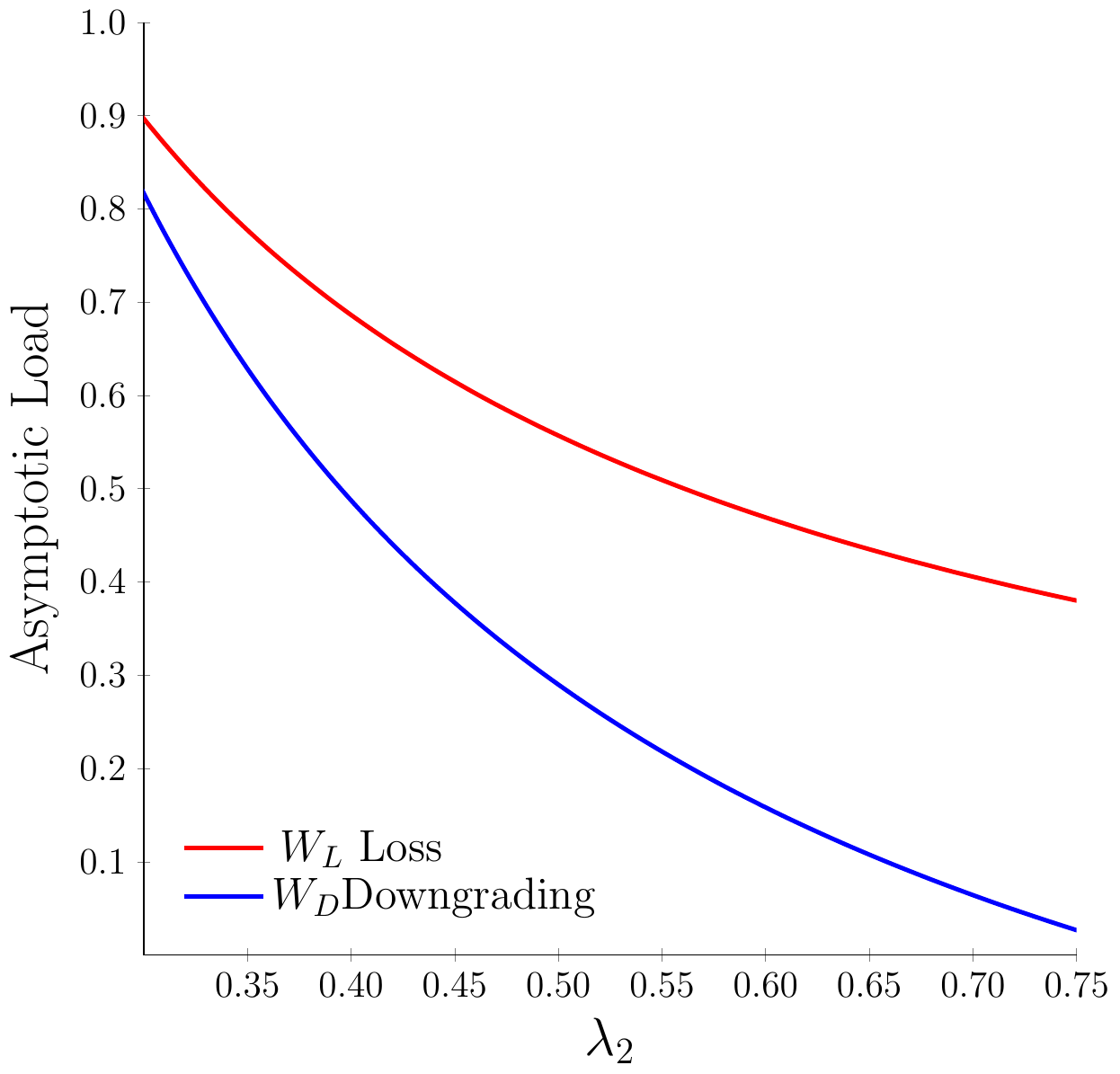}}
                \subfigure[{$J{=}3, A_2{=}2, A_3{=}3, \lambda_1{=}\lambda_2{=}0.2$} \newline{$(R)$ conditions{:} $\lambda_3{\in} (0.03,0.49)$}]{\includegraphics[width=.45\textwidth]{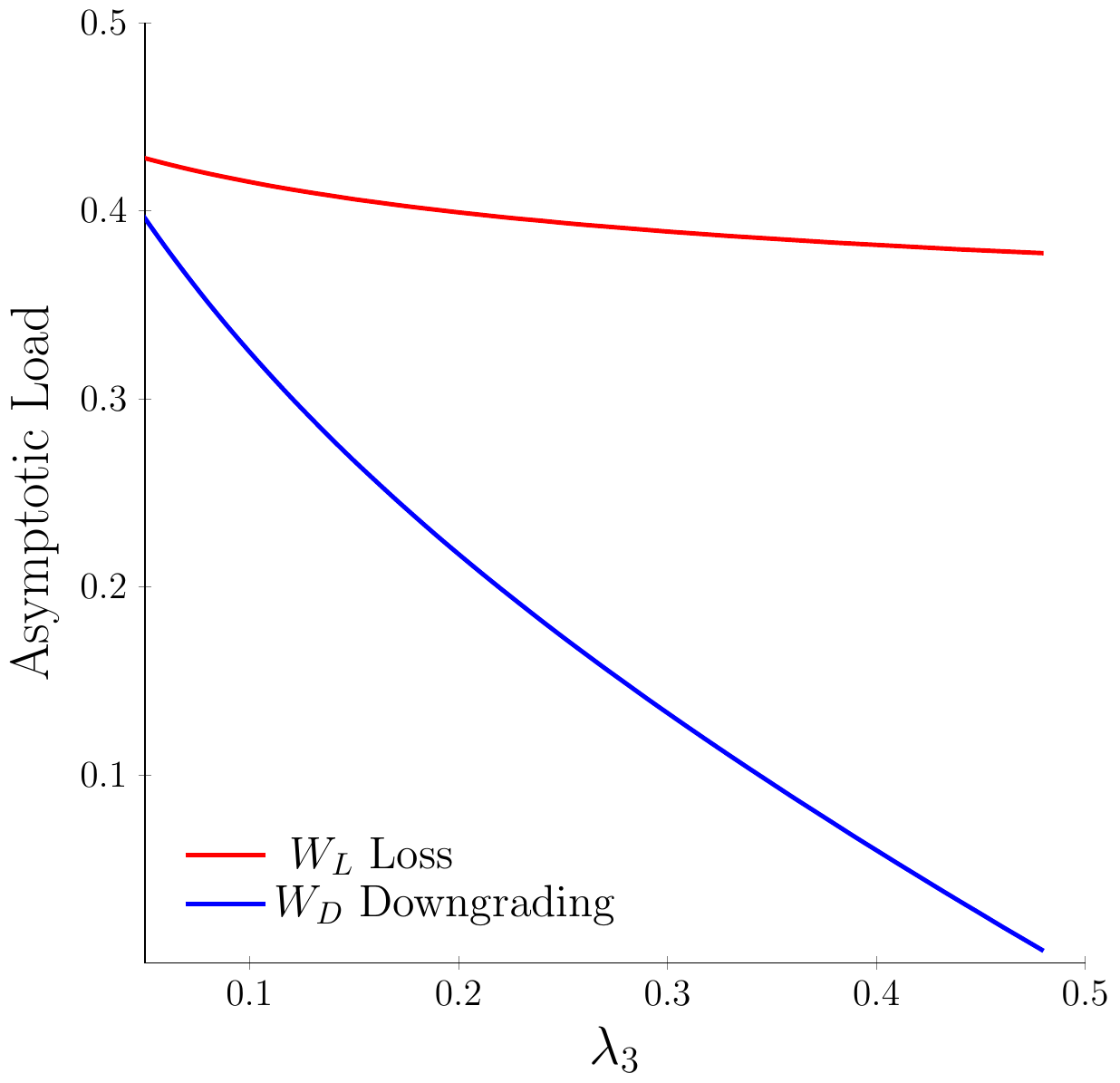}}
        \caption{Asymptotic load of non-downgraded/accepted requests with $A_1{=}1$, $c{=}1, c_0{=}0.99$ and all service rates equal to $1$.}\label{FigCCMP}
\end{figure}

\subsection*{Application to Video Transmission}
We  consider now a  link with large bandwidth, $10.0$ Gbps, in charge of  video streaming.  Requests that cannot be immediately served are lost. Video transmission is offered in two standard qualities, namely, \emph{Low Quality} (LQ) and \emph{High Quality} (HQ). From A\~{n}orga~\emph{et~al.}~\cite{Anorga2015},  the  bandwidth requirement   for YouTube's videos at 240p is  1485 Kbps,  and  for 720p it is 2737.27 Kbps. 

Using the values above, after renormalization, one takes $A_1{=}1$, $N=C^N{=}7061$ and $A_2 {=}2$,  $c=1$. Jobs arrive at rate $\lambda_2$ in this system asking for HQ transmission, but clients accept to watch the video in LQ. In particular $\lambda_1{=}0$. Service times are assumed to be the same for both qualities and taken as the unity, $\mu_1{=}\mu_2{=}1$.  Condition~\eqref{Cond} is satisfied when
\[
0.5<\lambda_2<1.
\]
We define  $C_0{=} \alpha C $, with $0{<}\alpha{<}1$. The quantity $\alpha_{\varepsilon}$ is defined as the largest value of $\alpha$ such that the loss  probability of a job is less than $\eps{>}0$. With the notations of Section~\ref{SecInv}, we write
\[
	\alpha_{\eps} = \sup \left\{\alpha\in(0,1) : \P\left(Y_{\ell^{*}}{+}C_0{>} C \right) < \eps \right\}.
        \]
Note that this is an approximation, since the variable $Y_{\ell^*}$ corresponds to the case when the scaling parameter $N$ goes to infinity. 
\begin{figure}[htpb!]
	\includegraphics[scale=0.5]{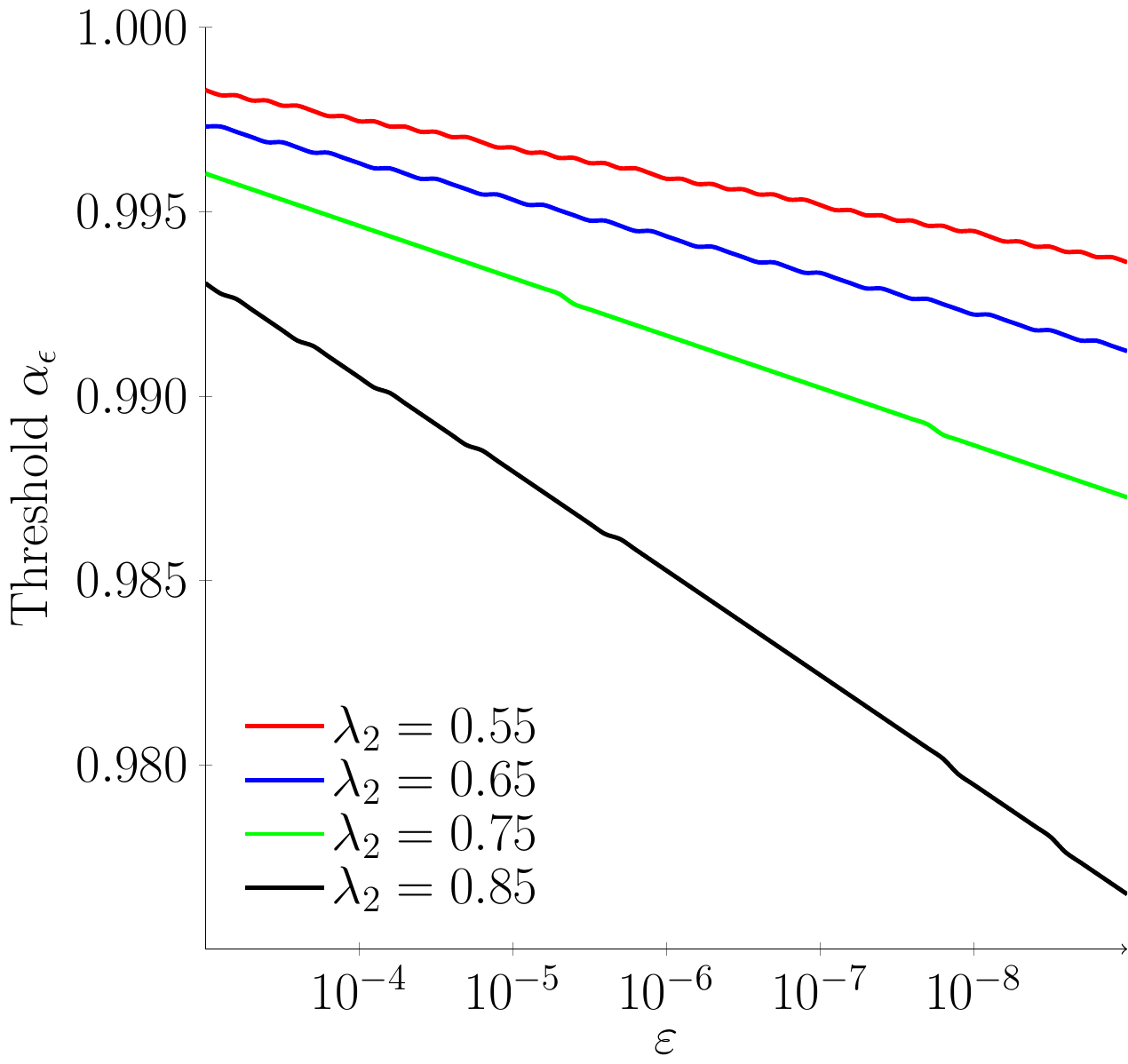}
	\caption{Maximal threshold for a loss probability equal to $\eps$.}\label{Fig_alpha}
\end{figure}
By using the explicit  expression of the distribution of $Y_{\ell^*}$ of Proposition~\ref{invariant}, Figure~\ref{Fig_alpha} plots the threshold $\alpha_\eps$ that  ensures a loss rate less than $\eps$ as a function of $\eps$, for several values of $\lambda_2$.  In the numerical example, taking $C_0{=}0.98 C$ is sufficient to get a  loss probability less than $10^{-7}$. 

Now let $\pi_\eps^-$ be the value of $\pi^-$ defined by Corollary~\ref{corol} for $C_0=\alpha_\eps C$. Recall that $\pi^-_\eps$ is the asymptotic equilibrium probability that a job is not downgraded is given by Relation~\eqref{pim},
\[
        \pi^{-}_{\eps} =  \frac{\alpha_\eps}{\lambda_2} - 1.
\] For comparison,  $\beta$ is defined as the corresponding acceptance probability  when no control is used in the system. We show in Figure~\ref{Fig_down} the relation between these quantities and the workload $\lambda_2$, for fixed  loss rates of $10^{-3}$, $10^{-6}$ and $10^{-9}$. We have $\beta {=} 1{-}1 {/}(2\lambda_2)$, see Robert~\cite[Proposition~6.19]{Robert}. The difference  $\beta{-}\pi^-$ can be seen as the fraction of jobs which are downgraded for our policy but lost in the uncontrolled policy. Intuitively it can be seen as the price of not rejecting any job. Notice also that the curves plotting $\pi_\eps^-$ for $\eps{=}10^{-3}$, $10^{-6}$, $10^{-9}$ are close and that  $\beta$ is larger than $\pi^-$. One remarks nevertheless  that, for  high loads, the system cannot hold these demands, because our policy is no longer effective. 

\begin{figure}[h!]
	\includegraphics[scale=0.5]{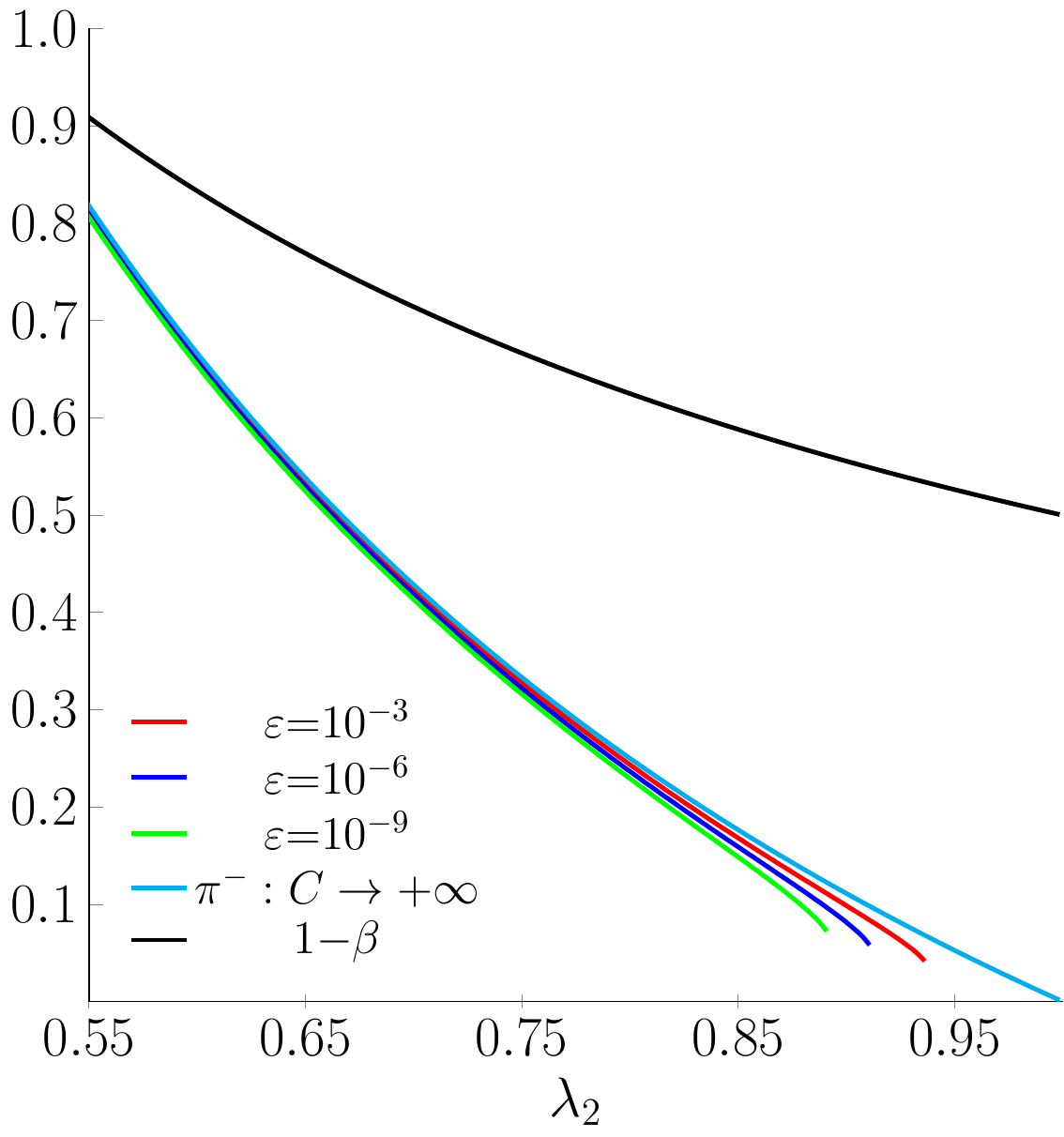}
	\caption{Fraction of non-downgraded jobs at equilibrium for the downgrading policy compared to the fraction of lost jobs in a pure loss system.}\label{Fig_down}
\end{figure}


\begin{thebibliography}{10}

\bibitem{Anorga2015}
Javier A{\~{n}}orga, Saioa Arrizabalaga, Beatriz Sedano, Maykel Alonso-Arce,
  and Jaizki Mendizabal.
\newblock {YouTube}{'}s {DASH} implementation analysis.
\newblock In {\em Proceedings of the 19th International Conference on
  Communications}, volume~50 of {\em Recent Advances in Electrical Engineering
  Series}, pages 61--66, Zakynthos, Ionion, Greece, 06 2015.

\bibitem{Arnold}
Vladimir~I. Arnol'd.
\newblock {\em Ordinary differential equations}.
\newblock Springer-Verlag, Berlin, 1992.
\newblock Translated from the third Russian edition by Roger Cooke.

\bibitem{Asmussen}
S{\o}ren Asmussen.
\newblock {\em Applied probability and queues}, volume~51 of {\em Applications
  of Mathematics}.
\newblock Springer-Verlag, New York, second edition, 2003.

\bibitem{Bean}
N.~G. Bean, R.~J. Gibbens, and S.~Zachary.
\newblock Asymptotic analysis of single resource loss systems in heavy traffic,
  with applications to integrated networks.
\newblock {\em Advances in Applied Probability}, 27(1):273--292, 1995.

\bibitem{Bean2}
N.~G. Bean, R.~J. Gibbens, and S.~Zachary.
\newblock Dynamic and equilibrium behavior of controlled loss networks.
\newblock {\em Annals of Applied Probability}, 7(4):873--885, 1997.

\bibitem{Doane}
David~P. Doane and Lori~E. Seward.
\newblock Measuring skewness: A forgotten statistic?
\newblock {\em Journal of Statistics Education}, 19(2):1--18, 2011.

\bibitem{MAMA}
Christine Fricker, Fabrice Guillemin, Philippe Robert, and Guiherme Thompson.
\newblock Analysis of downgrading for resource allocation.
\newblock {\em SIGMETRICS Performance Evaluation Review}, 44(2):24--26,
  September 2016.

\bibitem{FGRT}
Christine Fricker, Fabrice Guillemin, Philippe Robert, and Guilherme Thompson.
\newblock Analysis of an offloading scheme for data centers in the framework of
  {F}og computing.
\newblock {\em ACM Transactions on Modeling and Performance Evaluation of
  Computing System}, 1(4):16:1--12:18, 2016.

\bibitem{Gakhov}
F.~D. Gakhov.
\newblock {\em Boundary value problems}.
\newblock Dover Publications Inc., New York, 1990.
\newblock Translated from the Russian, Reprint of the 1966 translation.

\bibitem{icc13}
Fabrice Guillemin, Thierry Houdoin, and St\'ephanie Moteau.
\newblock Volatility of {YouTube} content in {Orange} networks and
  consequences.
\newblock In {\em Proceedings of {IEEE} International Conference on
  Communications, {ICC} 2013, Budapest, Hungary, June 9-13, 2013}, pages
  2381--2385, 2013.

\bibitem{itc25}
Fabrice Guillemin, Bruno Kauffmann, St\'ephanie Moteau, and Alain Simonian.
\newblock Experimental analysis of caching efficiency for {YouTube} traffic in
  an {ISP} network.
\newblock In {\em 25th International Teletraffic Congress, {ITC} 2013,
  Shanghai, China, September 10-12, 2013}, pages 1--9, 2013.

\bibitem{Hunt}
P.J. Hunt and T.G Kurtz.
\newblock Large loss networks.
\newblock {\em Stochastic Processes and their Applications}, 53:363--378, 1994.

\bibitem{Kelly:2}
F.P. Kelly.
\newblock Blocking probabilities in large circuit-switched networks.
\newblock {\em Advances in Applied Probability}, 18:473--505, 1986.

\bibitem{Kelly}
F.P. Kelly.
\newblock Loss networks.
\newblock {\em Annals of Applied Probability}, 1(3):319--378, 1991.

\bibitem{Robert}
Philippe Robert.
\newblock {\em Stochastic Networks and Queues}, volume~52 of {\em Stochastic
  Modelling and Applied Probability Series}.
\newblock Springer, New-York, June 2003.

\bibitem{Rudin}
Walter Rudin.
\newblock {\em Real and complex analysis}.
\newblock McGraw-Hill Book Co., New York, third edition, 1987.

\bibitem{Schwarz}
H.~Schwarz, D.~Marpe, and T.~Wiegand.
\newblock Overview of the scalable video coding extension of the {H.264/AVC}
  standard.
\newblock {\em IEEE Transactions on Circuits and Systems for Video Technology},
  17(9):1103--1120, September 2007.

\bibitem{sieber}
Christian Sieber, Tobias Hoßfeld, Thomas Zinner, Phuoc Tran-Gia, and Christian
  Timmerer.
\newblock Implementation and user-centric comparison of a novel adaptation
  logic for {DASH} with {SVC}.
\newblock In {\em IM'13}, pages 1318--1323, 2013.

\bibitem{Stolyar2}
Alexander~L. Stolyar.
\newblock An infinite server system with general packing constraints.
\newblock {\em Operations Research}, 61(5):1200--1217, 2013.

\bibitem{Stolyar}
Alexander~L. Stolyar.
\newblock Pull-based load distribution in large-scale heterogeneous service
  systems.
\newblock {\em Queueing Systems. Theory and Applications}, 80(4):341--361,
  2015.

\bibitem{vadlakonda}
S.~Vadlakonda, A.~Chotai, B.D. Ha, A.~Asthana, and S.~Shaffer.
\newblock System and method for dynamically upgrading / downgrading a
  conference session, April~6 2010.
\newblock US Patent 7,694,002.

\bibitem{Zachary}
Stan Zachary and Ilze Ziedins.
\newblock A refinement of the {H}unt-{K}urtz theory of large loss networks,
  with an application to virtual partitioning.
\newblock {\em The Annals of Applied Probability}, 12(1):1--22, 02 2002.

\bibitem{Zachary2}
Stan Zachary and Ilze Ziedins.
\newblock Loss networks.
\newblock In Richard~J. Boucherie and Nico~M. van Dijk, editors, {\em Queueing
  Networks}, volume 154 of {\em International Series in Operations Research \&
  Management Science}, pages 701--728. Springer US, 2011.

\end{thebibliography}
\end{document}